\documentclass[letterpaper, 10 pt, conference]{ieeeconf}  
\IEEEoverridecommandlockouts
\usepackage{balance}
\usepackage{color}
\usepackage{caption}
\usepackage{enumerate}
\usepackage{cite}
\usepackage{empheq}
\usepackage{mathrsfs}

\usepackage{mathtools}

\usepackage{tikz}
\usepackage{circuitikz}

\usepackage[font=small,labelfont=bf]{caption}

\renewcommand{\natural}{{\mathbb{N}}}

\newcommand{\integernonnegative}{\ensuremath{\mathbb{Z}}_{\ge 0}}
\newcommand{\real}{\ensuremath{\mathbb{R}}}

\newcommand{\until}[1]{\{1,\dots, #1\}}

\newcommand{\supscr}[2]{#1^{\textup{#2}}}

\newcommand{\diag}[1]{\operatorname{diag}(#1)}

\newcommand{\diam}{\operatorname{diam}}

\newcommand{\graph}{G}

\newcommand{\edges}{E}

\newcommand{\neighbors}{\mathcal{N}}

\renewcommand{\epsilon}{\varepsilon}

\newcommand{\argmin}{\ensuremath{\operatorname{argmin}}}
\newcommand{\argmax}{\ensuremath{\operatorname{argmax}}}

\usepackage{psfrag}

\usepackage{subfig}


\makeatletter
\pgfcircdeclarebipole{}{\ctikzvalof{bipoles/interr/height 2}}{spst}{\ctikzvalof{bipoles/interr/height}}{\ctikzvalof{bipoles/interr/width}}{

    \pgfsetlinewidth{\pgfkeysvalueof{/tikz/circuitikz/bipoles/thickness}\pgfstartlinewidth}

    \pgfpathmoveto{\pgfpoint{\pgf@circ@res@left}{0pt}}
    \pgfpathlineto{\pgfpoint{.6\pgf@circ@res@right}{\pgf@circ@res@up}}
    \pgfusepath{draw}
}

\def\pgf@circ@spst@path#1{\pgf@circ@bipole@path{spst}{#1}}
\tikzset{switch/.style = {\circuitikzbasekey, /tikz/to path=\pgf@circ@spst@path, l=#1}}
\tikzset{spst/.style = {switch = #1}}
\makeatother

\makeatletter
\let\proof\@undefined                        
\let\endproof\@undefined                  
\makeatother
\usepackage{graphicx,amssymb,amstext,amsmath,amsthm}

\usepackage{float}
\usepackage[bookmarks=true]{hyperref}
\usepackage{algorithm,algpseudocode}

\usepackage{multirow}
\usepackage{stfloats}

\renewcommand{\graph}{\mathcal{G}}
\newcommand{\nodes}{\mathcal{V}}
\renewcommand{\edges}{\mathcal{E}}
\newcommand{\interval}{\mathcal{I}}

\newtheorem{prop}{Proposition} 

\newtheorem{thm}{Theorem}
	\newtheorem{assumption}{Assumption}
\newtheorem{lem}{Lemma}
\newtheorem{defn}{Definition}

\newtheorem{problem}{Problem}

\let\oldReturn\Return
\renewcommand{\Return}{\State\oldReturn}

\usepackage{setspace}

\let\oldbibliography\thebibliography
\renewcommand{\thebibliography}[1]{%
  \oldbibliography{#1}%
}

\newcommand{\marginresolved}[1]{}

\newcommand{\scott}[1]{{\color{black} #1}}
\newcommand{\mmo}[1]{{\color{black} #1}}

\newcommand{\DIO}{\operatorname{DIO}}
\newcommand{\idl}{\underline{\operatorname{id}}}
\newcommand{\idh}{\overline{\operatorname{id}}}
\renewcommand{\argmin}{\ensuremath{\operatorname*{argmin}}}
\renewcommand{\argmax}{\ensuremath{\operatorname*{argmax}}}

\begin{document}

\title{\LARGE \bf Distributed Interval Observers for Bounded-Error LTI Systems} 

\author{%
  Mohammad Khajenejad$^*$, Scott Brown$^*$, and Sonia Mart{\'\i}nez \\
  \thanks{$^*$Equal contribution. M. Khajenejad, S. Brown and S. Mart{\'\i}nez are with the
    Mechanical and Aerospace Engineering Department of the Jacobs
    School of Engineering, University of California, San Diego, La
    Jolla, San Diego, CA, USA. (e-mail: \{mkhajenejad, sab007,
    soniamd\}@ucsd.edu.)}  \thanks{This work is partially funded by
    AFOSR and ONR grants FA9550-19-1-0235 and N00014-19-1-2471.}
}

\maketitle
\thispagestyle{empty}
\pagestyle{empty}

\begin{abstract}
  This paper proposes a novel distributed interval observer design for
  linear time-invariant (LTI) discrete-time systems subject to bounded
  disturbances. In the \mmo{proposed observer} algorithm, each agent in a networked group
  exchanges locally-computed framers or interval-valued state
  estimates with neighbors, and coordinates its update via an
  intersection operation. We show that the proposed framers are
  guaranteed to bound the true state trajectory of the system by
  construction, i.e., without imposing any additional assumptions or
  constraints. Moreover, we provide necessary and sufficient
  conditions for the collective stability of the distributed observer,
  i.e., to guarantee the uniform boundedness of the observer error
  sequence. In particular, we show that such conditions can be
  tractably satisfied through a constructive and distributed approach. \mmo{Moreover, we provide an algorithm to verify some structural conditions for a given system, which guarantee the existence of the proposed observer.}
  Finally, simulation results demonstrate the effectiveness of our
  proposed method
  compared to an existing distributed observer in the literature.
\end{abstract}

\section{Introduction}

Many large scale cyber-physical systems, such as electric power grids
\cite{zhao2019dse}, intelligent transportation systems
\cite{sun2017secure}, and industrial infrastructures
\cite{cheng2018industrial}, are equipped with sensor networks,
providing \textit{in situ} and diverse measurements to monitor
them. This makes possible the construction of system state estimates,
which are essential to guarantee the safe and effective operation of
these critical applications.  Motivated by this, an intense research
activity on the analysis and design of distributed estimation
algorithms has ensued. In this way, each sensor, equipped with local
communication and processing capabilities, interacts with neighboring
nodes to compute joint estimates cooperatively.

A way to obtain such estimates is to use a \textit{centralized
  observer}, by which a super node collects all measurements from the
nodes and fuses them in an optimal way. The ubiquitous Kalman filter
\cite{KalmanF.1960} and related approaches have been used extensively
for this purpose. However, these algorithms do not scale well as the
size of the network increases and are vulnerable to single-point
failures.  This spawned research on the design of \textit{distributed
  estimation filters} (for systems subject to known stochastic
disturbances) for sensor networks communicating only locally over a
possibly time-varying network \cite{rego2019est}. While these methods
are more scalable and robust to communication failures than their
centralized counterparts, they generally have comparatively worse
estimation error. An important class of algorithms that aim to
approach the estimation performance of their Kalman filter
counterparts, are Kalman-consensus filters, which combine a local
Kalman-like update with average consensus to align agents' estimates
\cite{saber2009kcf,kamal2013icf}. 
When stochastic characterization of disturbances is not available,
however, other techniques that leverage alternative information should
be considered.

In case the disturbances are known to be bounded, \textit{interval
  observers} are a popular method for obtaining robust, guaranteed
estimates of the state, due to their simplicity 
and computational efficiency
\cite{wang2015interval,mazenc2013robust,khajenejad2022interval_ACC}.
Hence, various approaches to design centralized interval observers for
various classes of dynamical systems have been proposed
\cite{tahir2021synthesis,khajenejad2021intervalACC,khajenejad2020simultaneousCDC,mazenc2011interval,chambon2016overview,cacace2014new,efimov2013interval}.
The main idea in most of the aforementioned designs is to synthesize
appropriate centralized observer gains to obtain a robustly stable and
positive observer error system for all realizations of the existing
uncertainties\cite{tahir2021synthesis,mazenc2011interval}. This
strategy, which usually boils down to solving centralized
semi-definite programs (SDP) subject to large numbers of constraints,
leads to theoretical and computational difficulties, and thus
infeasible solutions, especially for large-scale systems
\cite{chambon2016overview,cacace2014new,efimov2013interval}. In
addition to computational issues, the communication complexity of the
centralized approach does not scale well as the size of the network
increases. A recent study \cite{li2021ipr} proposes a distributed
interval observer for \textit{block-diagonalizable} linear
time-invariant (LTI) systems, which requires a certain structure on
the dynamics and the output of the system. Another work
\cite{li2021dos} designs an observer for LTI systems under
denial-of-service attacks. In addition,~\cite{wang2022robust} proposes
an internally positive representation (IPR)-based robust distributed
interval observer for continuous-time LTI systems.  {However,} the
proposed design relies on similarity transformations and the
satisfaction of certain algebraic constraints, which could lead to
moderately-performing results. Furthermore, all of the aforementioned
works use average consensus to share estimates throughout the network,
which limits the effectiveness of the proposed methods with respect to
time of convergence and estimation quality.

\textit{Contributions.}
To overcome the aforementioned drawbacks, this work
contributes to bridging the gap between interval observer design approaches and
distributed estimation algorithms in the presence of distribution-free
uncertainties. We introduce a novel method for synthesizing scalable
distributed interval observers for discrete-time LTI systems subject to bounded
additive disturbances. We provide necessary and sufficient
conditions for the  stability of our proposed observer. Our observer is correct by construction, and we leverage this correctness to
intersect interval estimates between neighboring nodes, ensuring that the
best possible (agent) estimate
is adopted by consensus in a finite number of
iterations. Furthermore,
we introduce the intuitive notion of ``collective positive detectability over neighborhoods'' (CPDN) which, \mmo{is sufficient} to
tractably compute gains that satisfy the aforementioned stability requirement in a distributed manner. This approach involves the solution to local and feasible linear programs (LP), which is potentially
less conservative and computationally more efficient than SDP-based approaches. \mmo{Finally, we provide an algorithm to verify if CPDN holds for a given system.}

\section{Preliminaries}

This section introduces basic notation, preliminary concepts and graph
theory notions used throughout the paper.

\textit{Notation.} Let $\real^n$, $\real^{n \times
  p}$, $\natural$, $\integernonnegative$, and $\mathbb{R}_{\geq 0}$
denote the $n$-dimensional Euclidean space, the sets of $n$ by $p$
matrices, natural numbers, nonnegative integers, and nonnegative real
numbers, respectively. For $M \in
\real^{n \times p}$, let $M_i$ and $M_{ij}$ denote the
$\supscr{i}{th}$ row of $M$, and the
$\supscr{(i,j)}{th}$ entry of $M$, respectively.
  Furthermore, for $M \in
  \real^{n \times p}$, we define $M^+ \in \real^{n \times p}$,
  such that $M^+_{ij} \triangleq \max \{M_{ij}, 0 \}$,
  $M^-\triangleq M^+-M$, and $|M|\triangleq M^++M^-$. In
addition, $M \succ 0$
(or, $M \succeq 0$, resp.) denote that $M$ is positive definite
(semi-definite, resp.), and $\rho(M)$ is used to denote the spectral
radius of $M$.  All the inequalities $\leq,\geq$ are considered
element-wise. As usual, $\mathbf{e}_i$ denotes the $\supscr{i}{th}$ vector of
the standard basis of $\real^n$. Finally, for $A^1,\dots,A^N \in
\real^{n \times n}$, $\diag{A^1,\dots,A^N} \in \real^{nN \times nN}$
denotes the block-diagonal matrix with block-diagonal elements
being $A^i,i\in \until{N}$.

\textit{Graph-theoretic Notions.}  Next, we recall some definitions
from Graph Theory.  A directed graph (digraph) $\graph$ is a set of
nodes $\nodes$ and a set of directed edges $\edges \subseteq \nodes
\times \nodes$. The set of neighbors of node $i$, denoted
$\neighbors_i$, is the set of all nodes $j$ for which there is an edge
$(i,j)\in\edges$.  We will assume that $i \in \neighbors_i$. A path
from node $i$ to node $j$ is a sequence of nodes starting with $i$ and
ending with $j$, such that any two consecutive nodes are joined by a
directed edge. The $d$-hop neighbors of agent $i$, denoted
$\neighbors_i^d$, is the set of nodes connected to $i$ by a path of
length no more than $d$.  The diameter of a graph is the largest
distance between any two nodes, i.e. $\diam\graph \triangleq
\max_{i,j}d(i,j)$, where $d(i,j)$ denotes the length of the shortest
path between $i$ and $j$.

\textit{Multi-dimensional Intervals.}  Finally, we introduce some
definitions and results regarding multi-dimensional intervals. A
(multi-dimensional) interval {$\interval \triangleq
  [\underline{s},\overline{s}] \subset \real^n$} is the set of all
vectors $x \in \real^n$ that satisfy $\underline{s} \le x \le
\overline{s}$.
\begin{prop}\cite[Lemma 1]{efimov2013interval}\label{prop:bounding}
  Let $A \in \real^{p \times n}$ and $\underline{x} \leq x \leq
  \overline{x} \in \real^n$. Then, $A^+\underline{x}-A^{-}\overline{x}
  \leq Ax \leq A^+\overline{x}-A^{-}\underline{x}$. As a corollary, if
  $A$ is non-negative, $A\underline{x} \leq Ax \leq A\overline{x}$.
\end{prop}

\section{Problem Formulation} \label{sec:Problem}
\noindent\textbf{\textit{System Assumptions.}}
Consider a multi-agent system \mmo{(MAS)} consisting of $ \nodes \triangleq
\until{N}$ agents, which interact over a time-invariant communication
graph $\graph = (\nodes, \edges)$. The agents are able to obtain
distributed measurements of a target as described by the following LTI
dynamics:
\begin{align}\label{eq:system}
\begin{array}{ll}
  \mathcal{P}:
  \begin{cases}
    x_{k+1} = Ax_k+Bw_k, &  \\
    y^i_k = C^ix_k+D^iv^i_k, & i \in \nodes, \quad k \in
    \integernonnegative,
  \end{cases}
\end{array}
\end{align}
where $x_k \in \real^{n}$ is the continuous state of the target system and
$w_k \in \interval_{w} \triangleq [\underline{w},\overline{w}] \subset
\real^{n_w}$ is bounded process disturbance. Furthermore, at time step
$k$, every agent $i \in \nodes$ takes a measurement
$y^i_k \in \real^{m_i}$, known only to itself, which is perturbed by
$v^i_k \in \interval_{v}^i \triangleq [\underline{v}^i,\overline{v}^i]
\subset \real^{n_v^i}$, a bounded sensor (measurement) noise
signal.
 Finally,
$A \in \real^{n \times n},B \in \real^{n \times n_w},C^i \in
\real^{m_i \times n}$ and $D^i \in \real^{m_i \times n^i_v}$ are system
matrices known to all agents.  The \mmo{MAS}'s goal is
to estimate the trajectories of 
\eqref{eq:system} in a distributed manner, when they are initialized
in an interval
$\interval_x \triangleq [\underline{x}_0,\overline{x}_0] \subset
\real^{n}$, with $\underline{x}_0,\overline{x}_0$ known to all agents.
Next, we define the notions of \textit{framers},
\textit{correctness}, and \textit{stability}, used throughout
the paper.

\begin{defn}[Framers]\label{defn:framers}
  For an agent $i\in\nodes$, the sequences $\{\overline{x}^i_k\}_{k
    \ge 0}$ and $\{\underline{x}^i_k\}_{k \ge 0} \subseteq \real^n$
  are called {upper} and {lower individual framers for (the
    state of) $\mathcal{P}$} if
    $ \underline{x}^i_k \leq x_k \leq \overline{x}^i_k, \quad$ for all $k
    \ge 0$.
  \mmo{Moreover,} we define the {individual lower} and {upper framer errors} \mmo{as follows:}
  \begin{align}\label{eq:error_1}
    & \underline{e}^i_k\triangleq x_k-\underline{x}^i_k, \quad
    \overline{e}^i_k \triangleq \overline{x}^i_k-x_k, \quad \forall k
    \ge 0.
  \end{align}
  \mmo{Given an MAS} 
  with target system $\mathcal{P}$ and
  communication graph $\graph$, a {distributed interval framer} is a distributed algorithm over $\graph$ that
  allows each agent $i$ to cooperatively compute upper and lower
  individual framers for
$\mathcal{P}$. Finally,
\begin{align}\label{eq:error_2}
  e_k \triangleq \begin{bmatrix}(\underline{e}^1_k)^\top &
    (\overline{e}^1_k)^\top & \cdots & (\underline{e}^N_k)^\top &
    (\overline{e}^N_k)^\top \end{bmatrix}^\top \in \real^{2Nn}.
\end{align}
is called the collective framer error, which is the vector of all individual lower and upper framer errors.
\end{defn}
\begin{defn}[Distributed Interval Observer]\label{defn:stability}
  A \emph{distributed interval framer} is input-to-state (ISS) stable
  if the collective framer error is bounded as follows:
  \begin{align*}
    \|e_k\| \leq
    \beta(\|e_0\|,k)+\mmo{\gamma}\big(\max_{0\le l \le k}|\Delta_l|\big),\quad \forall k \in
    \integernonnegative,
  \end{align*}
  where $\Delta_l \hspace{-.1cm}\triangleq [ w_l^\top \ (v^1_l)^\top \cdots
  (v^N_l)^\top]^\top\in\real^{n_w + Nn_v}$, and $\beta$ and {$\gamma$} are
  functions of classes $\mathcal{KL}$ and
  $\mathcal{K}_{\infty}$,
  respectively. 
  An ISS distributed interval framer is a
  {distributed interval observer.}
\end{defn}
The observer design problem can be stated as follows:
\begin{problem}\label{prob:SISIO}
  Given a multi-agent system and the LTI system in~\eqref{eq:system},
  design a distributed interval observer for~$\mathcal{P}$.
\end{problem}

\section{Proposed Distributed Interval Observer}
\label{sec:observer}
In this section, we describe our novel distributed interval observer
design, a necessary and sufficient condition for stability of the
proposed observer, and an LP-based distributed procedure for computing
stabilizing observer gains.

\subsection{Distributed Observer and its Correctness} \label{sec:obsv}
To address
Problem~\ref{prob:SISIO}, we propose a two-step distributed interval
framer (cf.~Definition~\ref{defn:framers}) for $\mathcal{P}$.  The
$\DIO$ Algorithm~\ref{alg1} provides a pseudocode description of our
observer, the details of which are further explained in this section \mmo{as follows}.

\noindent\textbf{\textit{ i) Propagation and Measurement Update:}}
At every $k+1 \in \integernonnegative$, given $\underline{x}^i_{k}$,
$\overline{x}^i_{k}$, $y^i_k$, and $y^i_{k+1}$, each agent~$i \in
\nodes$ performs a state propagation and a local measurement update
step using observer gains $L^i,\Gamma^i \in \real^{n \times m_i}$
which will be designed to satisfy desired observer properties:
\vspace{-.3cm}

{\small
\begin{gather}\label{eq:observer}
\begin{split}
   \hspace{-.1cm} \underline{x}^{i,0}_{k+1}
    &\hspace{-.1cm} =\hspace{-.1cm}  \tilde{A}^{i+}\underline{x}_k^i
        \hspace{-.1cm} -\hspace{-.1cm}  \tilde{A}^{i-}\overline{x}_k^i
       \hspace{-.1cm} +\hspace{-.1cm} (T^iB)^+\underline{w}
        \hspace{-.1cm}- \hspace{-.1cm}(T^iB)^-\overline{w}\hspace{-.1cm}+ \hspace{-.1cm}L^iy^i_k\hspace{-.1cm} +\hspace{-.1cm}  \Gamma^iy^i_{k+1} \\
        &
        \hspace{-.1cm} +\hspace{-.1cm}  ((L^iD^i)^+ \hspace{-.1cm} + \hspace{-.1cm} (\Gamma^iD^i)^+)\overline{v}
         \hspace{-.1cm} -\hspace{-.1cm}  ((L^iD^i)^- \hspace{-.1cm} +\hspace{-.1cm}  (\Gamma^iD^i)^-)\underline{v}
        , \\
  \hspace{-.1cm}  \overline{x}^{i,0}_{k+1}
    &\hspace{-.1cm} = \hspace{-.1cm} \tilde{A}^{i+}\overline{x}_k^i
        \hspace{-.1cm} -\hspace{-.1cm}  \tilde{A}^{i-}\underline{x}_k^i
       \hspace{-.1cm} + \hspace{-.1cm}(T^iB)^+\overline{w}
       \hspace{-.1cm} - \hspace{-.1cm}(T^iB)^-\underline{w}\hspace{-.1cm}+\hspace{-.1cm} L^iy^i_k\hspace{-.1cm} +\hspace{-.1cm}  \Gamma^iy^i_{k+1} \\
        &
        \hspace{-.1cm} +\hspace{-.1cm}  ((L^iD^i)^+\hspace{-.1cm}  +\hspace{-.1cm}  (\Gamma^iD^i)^+)\underline{v}
         \hspace{-.1cm} - \hspace{-.1cm} ((L^iD^i)^- \hspace{-.1cm} + \hspace{-.1cm} (\Gamma^iD^i)^-)\overline{v}
        ,
\end{split}
\end{gather}
}
where $T^i \triangleq I_n-\Gamma^iC^i$ and $\tilde{A}^i \triangleq
T^iA-L^iC^i$.
\mmo{Further}, $\underline{e}^{i,0}_k \triangleq x_k-\underline{x}^{i,0}_k$
and $\overline{e}^{i,0}_k \triangleq \overline{x}^{i,0}_k-x_k$
\mmo{are} the corresponding errors, and $e^0_k$ \mmo{is} the vector of all agents'
errors, as in \eqref{eq:error_2}.

\noindent\textbf{\textit{ ii) Network Update:}}
After the measurement update, each agent~$i$ \mmo{iteratively} shares its interval
estimate with its neighbors in the network, \mmo{and updates it by} taking the tightest
interval from all neighbors via intersection:  
\begin{align}\label{eq:network_update}
    \begin{split}
      \underline{x}^{i,t}_{k} = \max_{j\in\neighbors_i}\underline{x}^{j,t-1}_k,
      \quad \underline{x}^i_k=\underline{x}^{i,d}_{k}, \\
      \overline{x}^{i,t}_{k} = \min_{j\in\neighbors_i}\overline{x}^{j,t-1}_k,
      \quad \overline{x}^i_k=\overline{x}^{i,d}_{k},
    \end{split}
\end{align}
$\forall t \in \until{d}$, where $d \in \natural$ is
the number of network-update iterations.  Note that in case $d > 1$,
this iterative procedure computes the intersection of intervals with
the $d$-hop neighbors of each agent. \mmo{Consequently,}
each agent~$i$ obtains the \mmo{following} information:
\begin{align*}
  \underline{x}^i_k = \max_{j\in\neighbors_i^d}\underline{x}^{j,0}_k
  \quad \text{and} \quad \overline{x}^i_k =
  \min_{j\in\neighbors_i^d}\overline{x}^{j,0}_k.
\end{align*}
We will use this fact as a compact representation of the network
update~\eqref{eq:network_update}. Next, we show that the proposed
$\DIO$ algorithm constructs a distributed interval framer in the sense
of Definition \ref{defn:framers} for the plant $\mathcal{P}$.
\begin{algorithm}[H]
  \caption{Distributed Interval Observer ($\DIO$) at node $i$.}
  \label{alg1}
  \begin{algorithmic}[1]
    \renewcommand{\algorithmicrequire}{\textbf{Input:}}
    \renewcommand{\algorithmicensure}{\textbf{Output:}}
    \Require 
    $\underline{x}^i_0$, $\overline{x}^i_0$, $d$; \textbf{Output:}
    $\{\underline{x}^i_{k}\}_{k\ge0}$,
    $\{\overline{x}^i_{k}\}_{k\ge0}$; \State Compute $L^i$ $\Gamma^i$,
    and $T^i$ using \eqref{eq:stab_LMI}; \Loop \\
    $\quad \;\;$Compute $\underline{x}^{i,0}_{k+1}$ and
    $\overline{x}^{i,0}_{k+1}$ using \eqref{eq:observer};
    \For{$t = 1$  to $d$}
    \State Send $\underline{x}^{i,t-1}_{k+1}$ and
      $\overline{x}^{i,t-1}_{k+1}$ to $j \in \neighbors_i$;
     \State     $\displaystyle\underline{x}^{i,t}_{k+1} \gets\,
             \max_{j\in\neighbors_i}\underline{x}^{j,t-1}_{k+1};
             \,\,\overline{x}^{i,t}_{k+1} \gets\, \min_{j\in\neighbors_i}\overline{x}^{j,t-1}_{k+1};
           $
    \EndFor
        \State $\underline{x}^{i}_{k+1} \gets\,
          \underline{x}^{i,d}_{k+1}$; $\overline{x}^{i}_{k+1} \gets\,
          \overline{x}^{i,d}_{k+1}$; $k \gets k+1$;
     \EndLoop
    \Return $\{\underline{x}^i_{k}\}_{k\ge0}$, $\{\overline{x}^i_{k}\}_{k\ge0}$
  \end{algorithmic}
\end{algorithm}
\begin{lem}[Distributed Framer Construction]\label{lem:ind_framer}
  The $\DIO$ algorithm is a distributed interval framer
  for \eqref{eq:system}. 
\end{lem}

\begin{proof}
  From \eqref{eq:system} and the fact that $T^i =I_n-\Gamma^iC^i$, 
  we have
  \begin{align}\label{eq:interm_1}
    x_{k+1}\hspace{-.1cm}=\hspace{-.1cm}(\Gamma^iC^i\hspace{-.1cm}+\hspace{-.1cm}T^i)x_{k+1}\hspace{-.1cm}=\hspace{-.1cm}T^i(Ax_k\hspace{-.1cm}+\hspace{-.1cm}Bw_k)\hspace{-.1cm}+\hspace{-.1cm}\Gamma^iC^ix_{k+1}.
  \end{align}
  Plugging $C^ix_{k+1}=y^i_{k+1}-D^iv^i_{k+1}$ into, as well as adding
  the \emph{zero term} $L^i(y^i_k-C^ix_k-D^iv^i_k)$ to the right
  hand side of \eqref{eq:interm_1} results in
  {\small
    \begin{align}\label{eq:sys_equiv}
      x_{k+1}\hspace{-.1cm}=\hspace{-.1cm}\tilde{A}^ix_k\hspace{-.1cm}+\hspace{-.1cm}T^iBw_k\hspace{-.1cm}+\hspace{-.1cm}\Gamma^i(y^i_{k+1}\hspace{-.1cm}-\hspace{-.1cm}D^iv^i_{k+1})\hspace{-.1cm}+\hspace{-.1cm}L^i(y^i_k\hspace{-.1cm}-\hspace{-.1cm}D^iv^i_{k}).
    \end{align}
  }
  Applying Proposition \ref{prop:bounding} to all the uncertain terms
  in the right hand side of \eqref{eq:sys_equiv} shows that for each $i\in\nodes$,
  \begin{align*}
    \underline{x}^{i}_{k} \leq x_{k} \leq \overline{x}^{i}_{k}
    \implies \underline{x}^{i,0}_{k+1} \leq x_{k+1} \leq
    \overline{x}^{i,0}_{k+1},
  \end{align*}
  where $\underline{x}^{i,0}_{k+1},\overline{x}^{i,0}_{k+1}$ are given
  in \eqref{eq:observer}. This means that individual framers/interval
  estimates are correct. When the framer condition is satisfied for
  all nodes, the intersection of all the individual estimates of
  neighboring nodes (cf. \eqref{eq:network_update}) also results in
  correct interval framers, i.e.
  \begin{align*}
    \underline{x}^{i,0}_{k} \leq x_{k} \leq \overline{x}^{i,0}_{k},
    \; \forall i \in \nodes
    \implies \underline{x}^{i}_{k} \leq x_{k} \leq \overline{x}^{i}_{k},
    \; \forall i \in \nodes.
  \end{align*}
  Since the initial interval is known to all~$i$, by
  induction~\eqref{eq:observer}-\eqref{eq:network_update} constructs a
  correct distributed interval framer for~\eqref{eq:system}.
\end{proof}
\subsection{Collective Input-to-State Stability}\label{subsec:CISS}
In this subsection, we investigate conditions on the observer gains
$L^i$, $T^i$, and $\Gamma^i$, as well as the communication graph
$\graph$, that lead to an ISS distributed observer (cf. Definition
\ref{defn:stability}), which equivalently results in a uniformly
bounded observer error sequence $\{e_k\}_{k \ge 0}$ (given in
\eqref{eq:error_1}--\eqref{eq:error_2}), in the presence of bounded
noise.  For ease of exposition, in what follows, we ignore noise terms
and focus on asymptotic stability of the noiseless error dynamics,
which we will show implies collective input-to-state stability.

\noindent\textbf{\textit{Switched System Perspective.}}
We begin by stating a preliminary result that expresses the observer
error dynamics in the form of a specific switched system.
\begin{lem}
  \label{lem:switched}
  The collective error
  $e_k$ 
  has dynamics
  \begin{align}\label{eq:error-switched}
    e_{k+1} = H_k\hat A e_k,
  \end{align}
  where $\hat{A} \triangleq
  \diag{\hat{A}^1,\dots,\hat{A}^N}$,
  $\hat{A}^i\triangleq \begin{bmatrix}(\tilde{A}^i)^+ &
    (\tilde{A}^i)^- \\ (\tilde{A}^i)^- &
    (\tilde{A}^i)^+ \end{bmatrix}$.  $H_k \in \{0,1\}^{2Nn\times 2Nn}$
  is a binary matrix which selects a single minimizer or
  maximizer
    of the framers, i.e.,
  \begin{align}  \label{eq:H}
    \begin{split}
      (H_k)_{\idl(i,s),\idl(j^*,s)}
      = 1
      & \Leftrightarrow j^* = \min (\argmax_{j\in\neighbors_i^d} (\underline{x}^{j,0}_k)_s), \\
      (H_k)_{\idh(i,s),\idh(j^*,s)} =
      1
      & \Leftrightarrow j^* = \min (\argmin_{j\in\neighbors_i^d} (\overline{x}^{j,0}_k)_s), \\
    \end{split}
  \end{align}
  for $s \in \until{n}$ and $i\in\nodes$, where
  $\idl(i,s) = 2n(i-1)+s-1$ and
  $\idh(i,s) = 2n(i-1)+s+n-1$ encode the indices
  associated with the upper and lower framers of state $s$
  at node $i$. Furthermore, $H_k\hat A
  \in \mathcal{F} \subseteq \real^{2Nn\times 2Nn}$, where
  \begin{align*}
    \mathcal{F} &\triangleq \Big\{\begin{bmatrix} \underline
      a_{1,1}^\top & \overline a_{1,1}^\top &
       \underline a_{1,2}^\top & \overline a_{1,2}^\top & \cdots & \underline
      a_{N,n}^\top & \overline a_{N,n}^\top \end{bmatrix}^\top \\ &\quad \quad: \,
    \underline a_{i,s} \in \underline{\mathcal{F}}^i_{s}, \overline a_{i,s} \in \overline{\mathcal{F}}^i_{s}, \ s \in \until{n}, \ i\in \nodes \Big\} , \\
    \underline{\mathcal{F}}^i_{s} &\triangleq \left\{\mathbf{e}_j^\top
      \otimes [(\tilde A^j)^+ \ (\tilde A^j)^-]_s \in \real^{1\times 2Nn} \ : \
      {j\in\neighbors_i^d} \right\}, \\
    \overline{\mathcal{F}}^i_{s} &\triangleq \left\{\mathbf{e}_j^\top
      \otimes [(\tilde A^j)^- \ (\tilde A^j)^+]_s \in \real^{1\times 2Nn} \ : \ {j\in
        \neighbors_i^d} \right\}.
  \end{align*}
\end{lem}
\begin{proof}
  Notice that $\tilde{A}^i = \tilde{A}^{i+} - \tilde{A}^{i-}$
  in \eqref{eq:sys_equiv}. Rewriting in terms of the
  error by adding (or subtracting) \eqref{eq:sys_equiv} to
  \eqref{eq:observer} and \eqref{eq:network_update}, then setting the
  noise to zero, we obtain ${\color{magenta} e^0_{k+1}} = \hat A
  e_k$ and $e_k = H_k e^0_k$.
  Combining these yields \eqref{eq:error-switched}.
\end{proof}
Recall that the switching in \eqref{eq:error-switched} depends on the
state according to \eqref{eq:H} and always creates the smallest
possible error. In order to take the advantage of this property note that the set $\mathcal{F}$ has a specific structure known
as \emph{independent row uncertainty}, formally defined below.
\begin{defn}[Independent Row Uncertainty \cite{blondel2010jsr}]
  A set \mmo{of matrices} $\mathcal{M} \subset \real^{n\times n}$ {has} independent row uncertainties if
  \begin{align*}
    \mathcal{M} = \left\{\begin{bmatrix}a_1^\top & \cdots & a_n^\top\end{bmatrix}^\top
    \ : \ a_i \in \mathcal{M}_i, \ i \in \until{n}\right\},
  \end{align*}
  where all sets $\mathcal{M}_i \subset \real^{1\times n}$ are compact.
\end{defn}
\mmo{Next, we restate the following lemma on the spectral properties of
the sets with independent row uncertainties, that} will be used later in our stability analysis
of system \eqref{eq:error-switched}.
\begin{lem}\label{lem:lower-spectral-radius}\cite[Lemma 2]{blondel2010jsr}
 \mmo{Suppose} $\mathcal{M} \hspace{-.1cm} \subset \hspace{-.1cm} \real^{n\times n}$ \mmo{has}
  independent row uncertainties. \mmo{Then,} there exists $M_* \hspace{-.1cm}\in \hspace{-.1cm}\mathcal{M}$ such
  that
    $\rho(M_*) = \min_{M\in\mathcal{M}}\rho(M) =
    \lim\limits_{k\to\infty}\big{(}\min\limits_{M_i\in\mathcal{M},\\ i \in \until{k}}
    \|M_1\cdots M_k\|^\frac{1}{k}\big{)}.$
  The latter is known as the \emph{lower spectral radius} of \mmo{the set of matrices} $\mathcal M$.
\end{lem}
We can now state our first main stability result.
\begin{thm}[Necessary and Sufficient Conditions for Stability]\label{thm:stability}
  The error system \eqref{eq:error-switched} is globally exponentially
  stable \mmo{if and only if} there exists \mmo{$H_* \in \mathbb{R}^{2nN \times 2nN}$} which can be constructed according
  to \eqref{eq:H} for some values of $\underline x$ and
  $\overline{x}$, such that the matrix $H_*\hat A$ is Schur
  stable. Consequently, the $\DIO$ algorithm is ISS \mmo{if and only if such an $H_*$ exists}.
\end{thm}
\begin{proof}
  We first prove sufficiency and then necessity. Assume
  $H_*\hat{A}$ is Schur stable. Consider the comparison system $\tilde
  e_{k+1} = H_* \hat{A} \tilde e_k$ with initial condition $\tilde e_0
  = e_0$.  By the construction of
  $H_k$ in \eqref{eq:H}, which implies $H_*\hat{A}e_k \ge
  H_k\hat{A}e_k$, $\tilde e_k \ge e_k \ge 0$ for all $k \ge 0$ by
  induction.  Therefore, by comparison, \eqref{eq:error-switched} is
  globally exponentially stable.
  To prove
  necessity, note that \eqref{eq:error-switched} is 
  asymptotically stable only if the lower spectral radius of $\mathcal{F}$ is less
  than 1. By Lemma \ref{lem:lower-spectral-radius}, this implies
  existence of a stable $F_*= H_*\hat A$.
  Finally, having studied stability of the noise-free system, we now
  study the ISS property  of the noisy system:
\begin{align}\label{eq:error-noisy}
    e_{k+1} = H_k\hat A e_k + H_k(W_k + V_k),\
    \text{ where}
    \end{align}
    $W_k \hspace{-.1em}\triangleq [(\Lambda^{1}_k)^\top\dots (\Lambda^{N}_k)]^\top, \Lambda_k^i \triangleq \begin{bmatrix}(T^i B)^+\underline{s}_k + (T^i B)^-\overline{s}_k \\
      (T^i B)^-\underline{s}_k + (T^i
      B)^+\overline{s}_k \end{bmatrix},$ with
    $\underline{s}_k \triangleq w_k - \underline{w}$,
    $\overline{s}_k \triangleq \overline{w} - w_k$ and $V_k$ is defined
    similarly to $W_k$, with $(L^iD^i)^* + (\Gamma^iD^i)^*$ replacing
    $(T^iB)^*$, for $*\in\{+,-\}$ and $\underline{v}^i$,
    $\overline{v}^i$, and $v^i_k$ replacing $\underline{w}$, $\overline{w}$, and $w_k$
    respectively.
  As before, we can use the comparison system
  \begin{align}\label{eq:iss-comp}
    \tilde e_{k+1} = H_*\hat A \tilde e_k + H_*(W_k + V_k), \quad \tilde e_0 = e_0
  \end{align}
  It is well known that stable LTI systems are ISS \cite{sontag2008input}. Again, \eqref{eq:H} guarantees $\tilde e_k
  \ge e_k \ge 0 \ \forall k \ge 0$ by induction, regardless of the values of $W_k$ and $V_k$. By this comparison, the ISS property of the system
  \eqref{eq:iss-comp} implies that \eqref{eq:error-noisy} is ISS.
\end{proof}
Theorem~\ref{thm:stability} is only an \emph{existence} result.  It
does not provide a method for constructing $H_*$, which could be a
difficult combinitorial problem. Therefore, in the next section we
provide \mmo{a tractable approach} that allows for the
computation of stabilizing gains and the corresponding $H_*$ \mmo{in Theorem~\ref{thm:stability}}.

\subsection{Distributed Stabilizing of the Error Dynamics}
In this subsection, we show that the ISS property formalized in
Theorem~\ref{thm:stability} can be tractably verified in a
constructive and distributed manner.  The approach is motivated by the
representation~\eqref{eq:error-switched}, in which $H_k$ exchanges
rows of $\hat A$ to achieve the best possible estimate. The main
idea is that each agent, depending on its observation of the system,
contributes to stabilizing the state trajectory in \emph{some}, not
necessarily all, dimensions.  To guarantee stability using our design
approach, we will use the following assumption, which characterizes
the interplay between the agents' local observations and their
communication over the network.
\begin{assumption}[Collective Positive Detectability over
  Neighborhoods (CPDN)]\label{ass:col_pos_det_neigh}
  There is a $d^*\in\natural$ such that for every state
  $s \in \until{n}$ and every agent $i\in\nodes$, there is an agent
  $\ell(i,s)\in\neighbors_i^{d^*}$ such that there exist gains
  $T^{\ell(i,s)}$, $L^{\ell(i,s)}$, and $\Gamma^{\ell(i,s)}$
  satisfying
  $\|(T^{\ell(i,s)}A-L^{\ell(i,s)}C^{\ell(i,s)})_s\|_1 < 1$.
\end{assumption}
\mmo{Assumption~\ref{ass:col_pos_det_neigh}} captures a broad range
of conditions on the system and graph structure that can result in a stable observer.
\mmo{It} is \mmo{also}
different from the well-known notion of ``collective observability"
\cite{wang2022robust}, which is not sufficient here due to the non-negativity of the
error dynamics in the interval observer design.
Furthermore, \mmo{the conditions in Assumption~\ref{ass:col_pos_det_neigh}}, can be
verified using a distributed procedure described in Algorithm~\ref{alg2}, \mmo{summarized here.} First, each node $i\in\nodes$
independently solves the \mmo{following} linear program:
\begin{align}\label{eq:stab_LMI}
  \begin{array}{rl}
    &\min\limits_{\{E^i,L^i,T^i,\Gamma^i\}} \quad \sum_{j=1}^n \sum_{t=1}^n E^i_{jt} \\
    &\mathrm{s.t.}\  -E^i \leq T^iA-L^iC^i \leq E^i,  T^i = I_n - \Gamma^iC^i.
  \end{array}
\end{align}
Then, nodes exchange their $\tilde A^i$ matrices with increasingly
larger neighborhoods until stabilizing agents are found for every state.
\mmo{The following lemma formalizes this result}.
\begin{algorithm}
  \caption{$\DIO$ initialization at node $i$.}
  \label{alg2}
  \begin{algorithmic}[1]
    \renewcommand{\algorithmicrequire}{\textbf{Input:}}
    \renewcommand{\algorithmicensure}{\textbf{Output:}}
    \Require $A$, $C^i$, $\neighbors_i$;
    \textbf{Output:} $L^{i}$, $T^{i}$, $\Gamma^{i}$, $d^*$
    \State Compute $L^{i,*}$ $\Gamma^{i,*}$, and $T^{i,*}$ by solving \eqref{eq:stab_LMI};
    \State $\mathcal{Q}_i \gets \{T^{i,*}A - L^{i,*}C^{i,*}\}$; $d^* \gets 1$;
    \While{$d^* \le \diam \graph$}
    \If{$\forall s \in \until{n}$, $\exists P \in \mathcal{Q}_i$ s.t. $\|(P)_s\|_1 < 1$}
    \State \textbf{break}
    \EndIf
    \State Send $\mathcal{Q}_i$ to $j\in\neighbors_i$
    and recieve $\mathcal{Q}_j$ from $j\in\neighbors_i$;
    \State $\mathcal{Q}_i \gets \bigcup_{j\in\neighbors_i} Q_j$;
     $d^* \gets d^* + 1$;
    \EndWhile
    \State \textbf{for} {$t = 1$ to $\diam\graph$} \textbf{do} $d_i^* \gets \max_{j\in\neighbors_i}d_j^*$ \textbf{end for}
    \Return $L^{i,*}$, $T^{i,*}$, $\Gamma^{i,*}$, $d_i^*$
  \end{algorithmic}
\end{algorithm}

\begin{lem}
Assumption~\ref{ass:col_pos_det_neigh} holds if and only if Algorithm~\ref{alg2} returns $d^* \leq \diam\graph$.
\end{lem}
\begin{proof}
  Assume Assumption~\ref{ass:col_pos_det_neigh} does not hold for some
  agent $i$. Then the condition in line 4 of Algorithm~\ref{alg2} will
  never be met, resulting in $d_i^* = \diam{\graph} + 1$. After the
  max consensus on line 10, all agents will return $d^*=\diam\graph + 1$. On the other hand, if
  Assumption~\ref{ass:col_pos_det_neigh} holds, the condition in line
  4 will be met after less than $\diam{\graph}$ iterations for every node.
\end{proof}
\mmo{Finally, we show that the solutions to the LP in \eqref{eq:stab_LMI} are the corresponding stabilizing observer gains.}
\begin{thm}[Distributed Interval
  Observer Design]\label{thm:suff_stability}
  Suppose Assumption \ref{ass:col_pos_det_neigh} holds.  Then the
  $\DIO$ algorithm is ISS with $d = d^*$ and the corresponding observer
  gains ${L}^{*,i}$, ${T}^{*,i}$, and ${\Gamma}^{*,i}$ that are
  solutions to \eqref{eq:stab_LMI}.
\end{thm}
\begin{proof}
  We will construct $H_*$, which by Theorem \ref{thm:stability} is
  sufficient for the ISS property to hold. For each node $i\in\nodes$
  and state $s \in \until{n}$, using $\ell(i,s)$ from Assumption
  \ref{ass:col_pos_det_neigh},
\begin{align*}
    (H_*)_{\idh(i,s),\idh(\ell(i,s),s)} = 1, \  \
    (H_*)_{\idl(i,s),\idl(\ell(i,s),s)} = 1,
\end{align*}
and all other entries are zero. Since $\ell(i,s) \in
\neighbors_i^{d^*}$, 
then $H_*$ could be constructed according to \eqref{eq:H} for some
$\underline{x}$ and $\overline{x}$. With $H_*$ defined as such, rows
$\idl(i,s)$ and $\idh(i,s)$ of
$H_*\hat A$ are equal to rows $\idl(\ell(i,s),s)$
and $\idh(\ell(i,s),s)$ of $\hat A$,
respectively (cf. Lemma~\ref{lem:switched}). From the definition of $\hat A^i$, it is clear that
$\|(\tilde A^i)_s\|_1$ = $\|(\hat A^i)_s\|_1$ = $\|(\hat
A^i)_{s+n}\|_1$. 
  Note that the gains $T^i$ and $L^i$ are computed
by \eqref{eq:stab_LMI}, which independently minimize the $1$-norm of
each row of $\tilde A^{i}$, since the $\supscr{s}{th}$ rows of $T^i$ and $L^i$
only affect the $\supscr{s}{th}$ row of \scott{$\tilde A^i \triangleq T^iA-L^iC^i$}. Moreover, Assumption
\ref{ass:col_pos_det_neigh} guarantees $\|(\tilde A^{\ell(i,s)})_s\|_1
< 1$ for each $s$. All of this implies $ \|(H_*\hat
A)_{\idl(i,s)}\|_1 = \|(H_*\hat
A)_{\idh(i,s)}\|_1 < 1$.
Since this holds for every row of the matrix $H_*\hat A$, then $\rho(H_*\hat A) \leq \|H_*\hat A\|_\infty \triangleq \max\limits_{1\leq i \leq 2nN} \sum_{s=1}^{2nN} |(H_*\hat A)_{ij}|< 1$. 
\end{proof}
\vspace{-.15cm}
\scott{Even if Assumption \ref{ass:col_pos_det_neigh} doesn't hold,
  or if the algorithm is run with $d < d^*$, it is possible that
  solving \eqref{eq:stab_LMI} will result in stabilizing gains which can
  be verified by Theorem~\ref{thm:stability}.}

\section{Numerical Examples}\label{sec:example}
This section demonstrates the effectiveness of our distributed
interval observer applied to two LTI target systems. \\
\noindent\textit{Example 1:}
Consider a system from \cite{li2021ipr} in the form of \eqref{eq:system}:
\begin{align}\label{eq:ex_1}
\begin{array}{c}
    N = 3,\\
    n = 4, \\
    B = I_4,
\end{array}
A = \begin{bmatrix}
1 & 0 & 1 & 0 \\
0 & 1 & 0 & 1 \\
0 & 0 & 1 & 0 \\
0 & 0 & 0 & 1
\end{bmatrix},
\begin{array}{c}
    C^1 = \begin{bmatrix}1 & 0 & 0 & 0\end{bmatrix}, \\[3pt]
    C^2 = \begin{bmatrix}0 & 1 & 0 & 0\end{bmatrix}, \\[3pt]
    C^3 = \begin{bmatrix}1 & 1 & 0 & 0\end{bmatrix}.
\end{array}
\end{align}
The network is fully connected, so $d=d^*=1$. The initial state $x_0$ and the process noise $w_k$ are bounded by
$\underline{x}_0 = \begin{bmatrix} -20 & -15 & -0.5 &
  0\end{bmatrix}^\top,\overline{x}_0 = \begin{bmatrix} 10 & 25 & 2 &
  3\end{bmatrix}^\top, \underline{w} = \begin{bmatrix}-0.1 & -0.1 & -1
  & -1\end{bmatrix}^\top$ and $\overline{w} = \begin{bmatrix}0.1 & 0.1
  & 1 & 1\end{bmatrix}^\top$, while $v^i_k$ is bounded by
$\underline{v}^i = -i$, $\overline{v}^i = i$, for $i=1,2,3$.  To show
the advantage of including a network update in the observer design, we
also implemented individual local observers for each agent $i$,
i.e., using only~\eqref{eq:observer} and local output $y^i_k$. The
gains of the local observers are the same as in the distributed
observer. Using CVX~\cite{cvx,Grant.08} to solve the design LPs in
\eqref{eq:stab_LMI}, we obtained gains for each node:

\vspace{-.35cm}
{\small
\begin{gather*}
    L^1 \hspace{-.1cm}=\hspace{-.1cm} \begin{bmatrix} 0 & 0 & -1 & 0\end{bmatrix}^\top\hspace{-.15cm},
    L^2 \hspace{-.1cm}=\hspace{-.1cm} \begin{bmatrix} 0 & 0 & 0 & -1\end{bmatrix}^\top\hspace{-.15cm},
    L^3 \hspace{-.1cm}=\hspace{-.1cm} -\begin{bmatrix} 0 & 0 & .5& .5\end{bmatrix}^\top\hspace{-.15cm},\\
    \Gamma^1 = \begin{bmatrix}1 & 0 & 1 & 0\end{bmatrix}^\top\hspace{-.15cm},
    \Gamma^2 = \begin{bmatrix} 0 & 1 & 0 & 1\end{bmatrix}^\top\hspace{-.15cm},
    \Gamma^3 = \begin{bmatrix}.5& .5& .5 & .5\end{bmatrix}^\top\hspace{-.15cm}, \\
T^1 \hspace{-.1cm}=\hspace{-.15cm} \begin{bmatrix}
    0      & 0 & 0 & 0 \\
    0 & 1 & 0 & 0 \\
   -1 & 0 & 1 & 0 \\
    0      & 0 & 0 & 1
    \end{bmatrix}\hspace{-.1cm},T^2 \hspace{-.1cm}=\hspace{-.15cm} \begin{bmatrix}
    1  &  0  & 0 & 0 \\
    0  &  0  & 0 & 0 \\
    0  &  0  & 1 & 0 \\
    0  &  -1  & 0 & 1
    \end{bmatrix}\hspace{-.1cm},T^3 = \begin{bmatrix}
    .5 & -.5 &       0   &      0 \\
   -.5 &  .5 &       0   &      0 \\
   -.5 & -.5 &  1        & 0 \\
   -.5 & -.5 &       0   & 1 \\
    \end{bmatrix}.
\end{gather*}
}
Figure~\ref{fig:ex1} shows the resulting framer trajectories
and interval widths for each observer. All observers maintain a \emph{correct} interval estimate, but some framers of the individual observers are clearly unstable. On the other hand, the distributed observer is able to maintain a tight interval around the true state trajectory, despite the unstable individual observers.\\
\noindent\textit{Example 2:}
Consider the following multi-agent system (a discretized version of
the example from \cite{wang2022robust}, with $\Delta_t = 0.01$) in the
form of \eqref{eq:system} with $N=6$, $n=12$, $A = I_6
\otimes \begin{bmatrix}
  1 & 0.01  \\
  0 & 1 \end{bmatrix}$, $B=I_6 \otimes \begin{bmatrix}
  0.01 & 0.0001  \\
  0 & .01 \end{bmatrix}$, and $C^i =
\mathbf{e}_i\otimes \begin{bmatrix}1 & 0 \end{bmatrix}$. The
communication graph $\graph$ is a directed ring with nodes $\nodes =
\until{6}$ and edges $\edges = \{(i, (i+1 \bmod{6})) \,:\,
i\in\nodes\}$. Moreover, the initial state is $x_0 =
\mathbf{1}_6\otimes \begin{bmatrix} 0.7032 & 0.0457 \end{bmatrix}$ and
the process noise is given by $w_k = \mathbf{1}_4
\otimes \begin{bmatrix} \frac{1}{2}\sin(0.01k) & 0 &
  \frac{1}{5}\cos(0.01k)\end{bmatrix}^\top$, $\overline{w} =
\mathbf{1}_4 \otimes \begin{bmatrix}\frac{1}{2} & 0 &
  \frac{1}{5}\end{bmatrix}^\top$, $\underline{w} = \mathbf{1}_4
\otimes \begin{bmatrix}-\frac{1}{2} & 0 &
  -\frac{1}{5}\end{bmatrix}^\top$.  The measurement noise is given by
$v^1_k = v^4_k = \frac{1}{5}\sin^2(0.01k)$, $v^2_k = v^5_k = 0$, and
$v^3_k = v^6_k = \frac{1}{3}\cos(0.02k)$, with bounds $\overline{v}^1
= \overline{v}^4 = \frac{1}{5}$, $\overline{v}^2 = \overline{v}^5 =
0$, $\overline{v}^3 = \overline{v}^6 = \frac{1}{3}$, $\underline{v}^1
= \underline{v}^2 = \underline{v}^4 = \underline{v}^5 = 0$, and
$\underline{v}^3 = \underline{v}^6 = -\frac{1}{3}$. The $\DIO$ gains
for each node $i \in \nodes$ are
    $L^i = \mathbf{e}_i \otimes [0 \ -100]^\top, \quad
    \Gamma^i = \mathbf{e}_i \otimes [1 \ 100]^\top,
    T^i = I_{12} + \mathbf{e}_i\mathbf{e}_i^\top \otimes \begin{bmatrix}-1 & 0 \\ -100 & 0\end{bmatrix}$.
    Assumption~\ref{ass:col_pos_det_neigh} is satisfied with
    $d^*=5$, but $d=1$ still results in a stable observer.
    Figure~\ref{fig:ex2} shows the state trajectories and interval
    widths from the DIO algorithm \scott{with $d=1$}.  For comparison, we also plot the
    results of the continuous time distributed observer proposed in
    the work \mmo{in} \cite{wang2022robust}.  As can be seen, our observer is
    able to maintain tighter intervals around the true states, which
    can be due to the advantage of using the equivalent system
    representation \eqref{eq:sys_equiv}, solving LPs instead of SDPs
    and/or leveraging a more general observer structure.
\section{Conclusion and Future Work} \label{sec:conclusion}
In this paper, a novel distributed interval observer synthesis was
introduced for linear time-invariant discrete-time systems that are
subject to bounded noise. At each time step, agents first individually
compute upper and lower framers for the system state, based on their
own observation of the system. Then, they communicate to their
neighbors, exchanging and updating their estimates iteratively by
selecting the tightest intervals among their agents' estimates. The
true state trajectory of the system was guaranteed to be bounded by
the updated framers. Furthermore, necessary and
sufficient conditions assuring collective stability were provided. Such conditions were shown to be collectively achievable by contribution of all agents through solving tractable linear programs in a distributed manner. \mmo{In addition, we provided an algorithm to verify if such sufficient conditions hold for a given system.}
 Finally,
the advantage of the proposed distributed design compared to its
centralized version and also versus a benchmark observer in the literature was demonstrated through simulation examples. In
future work, extension to nonlinear settings with
time-varying networks, and the
effect of adversarial
agents and communication and measurement failures will be considered.

\begin{figure}[t!]
  \centering
  {\includegraphics[width=0.98\columnwidth]{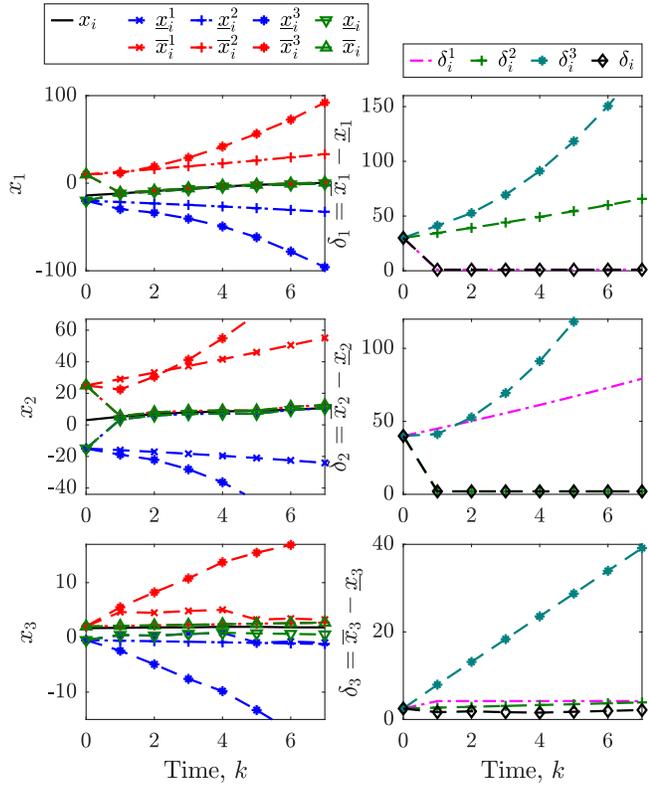}}
  \caption{{\small States $x_1$ through $x_3$ (black), along with the upper (red) and lower (blue) framers (left) and interval widths $\delta_i \triangleq \overline{x}_i- \underline{x}_i$ (right) from the observers, for system \eqref{eq:ex_1}.
  Superscripts \mmo{$1$--$3$} denote framers from individual agents \mmo{$1$--$3$}, respectively, while the green marker (left) denotes the resulting estimate by DIO.}}
  \label{fig:ex1}
\end{figure}
\begin{figure}[t!]
  \centering
  {\includegraphics[width=0.98\columnwidth]{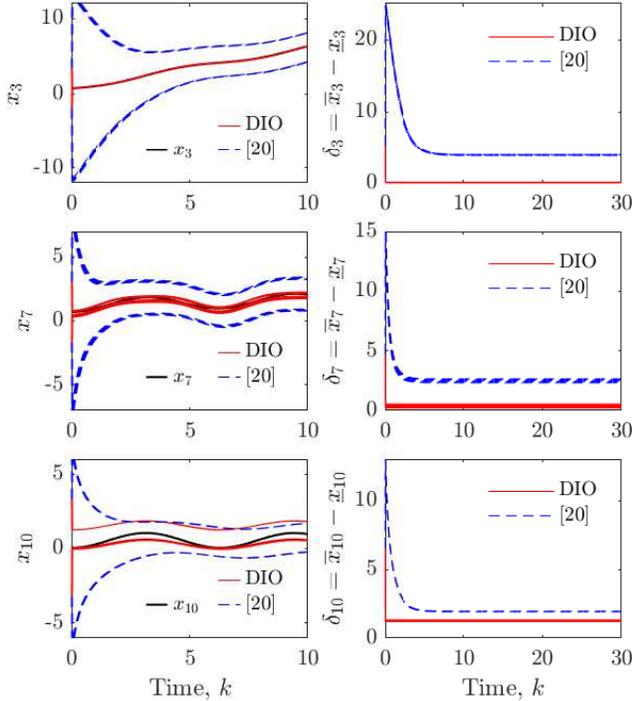}}
  \caption{{\small State framers and interval widths from all agents for
    Example 2, from $\DIO$ (red) and the observer from \cite{wang2022robust} (blue).}}
  \label{fig:ex2}
\end{figure}

\bibliographystyle{unsrturl}
{\tiny
  \bibliography{biblio}
}

\end{document}